\documentclass[11pt,journal,draftcls,letterpaper,onecolumn]{IEEEtran}

\usepackage{amsmath}
\usepackage{amssymb}  
\usepackage{amsthm}
\usepackage{float}
\usepackage{rotating} 
\usepackage{graphicx}
\usepackage{subfigure} 
\usepackage{epsfig}
\usepackage{epstopdf}
\usepackage{cite}
\usepackage{array}
\usepackage[table]{xcolor}
\usepackage{color}
\usepackage{booktabs}
\usepackage{multirow}
\usepackage{pdflscape}
\newcommand{\figref}[1]{Fig.~\ref{#1}}

\newtheorem{lemmacounter}{Theorem}
\newtheorem{Proposition}[lemmacounter]{Proposition}

\newcommand{\comment}[1]{}

\IEEEoverridecommandlockouts

\begin{document}

	\title{Reducing the Mutual Outage Probability of Cooperative Non-Orthogonal Multiple Access}

	\author{\IEEEauthorblockN{Sana Riaz, Fahd Ahmed Khan, Sajid Saleem and Qasim Zeeshan Ahmed} 
		\thanks{
			Sana Riaz and Fahd Ahmed Khan are with the School of Electrical Engineering and Computer Science (SEECS), National University of Sciences and	Technology (NUST), Islamabad, Pakistan. Email: \{14mseesriaz, fahd.ahmed\}@seecs.edu.pk.  Sajid Saleem is with Department of Computer and Network Engineering, University of Jeddah, Jeddah, Saudi Arabia. Email: ssaleem@uj.edu.sa. Qasim Zeeshan Ahmed is with the School of Computing and Engineering at the University of Huddersfield, United Kingdom. Email: q.ahmed@hud.ac.uk. The corresponding author is Fahd Ahmed Khan.} }
	
	\maketitle
	
	\begin{abstract} 
		In this letter, a new power allocation scheme is proposed to improve the reliability of cooperative non-orthogonal multiple access (CO-NOMA). The strong user is allocated the maximum power, whereas the weak user is allocated the minimum power. This power allocation alters the	decoding sequence along with the signal-to-interference plus	noise ratio (SINR), at the users. The weak user benefits from receiving multiple copies of the signal whereas the strong user benefits from the higher power allocation. {Numerical simulation results show that the proposed scheme has a lower mutual outage probability (MOP) and offers better reliability as compared to the conventional power allocation scheme for CO-NOMA. An exact closed-form expression of MOP is derived	for the two-user CO-NOMA system and it is shown that each user achieves full diversity. The proposed allocation is able to achieve approximately 30\% higher transmission rate at 15 dB as compared to conventional CO-NOMA in a practical non-power balanced scenario. 
			
		}
	\end{abstract}  
	
	
	\section{Introduction}

	Non-orthogonal multiple access (NOMA) is a promising scheme that offers
	higher throughput by allowing each user in the network to operate at the same time and frequency~\cite{BibMakki2020, Application, BibDai2015}. In particular, NOMA allocates different power levels to the users and transmits a superimposed message, which is then decoded by the intended users through successive interference cancellation (SIC)~\cite{BibDing2014}.
		The performance of NOMA, in terms of the outage probability and ergodic sum-rate, in a downlink cellular network with randomly deployed mobile users has been investigated in \cite{Performance}, demonstrating that NOMA can achieve a higher rate than orthogonal multiple access (OMA). Outage performance of NOMA in presence of imperfect channel knowledge was analyzed in \cite{Partial}. The impact of user pairing on the performance of two NOMA systems, NOMA with fixed power allocation and cognitive radio inspired NOMA, was discussed in \cite{UserPairing}.
		Performance gain of both schemes depends on different kinds of user pairing, CR-NOMA prefers to pair the user with the best channel
		condition with the user having second best channel condition, whereas the user with the worst channel condition is preferred by F-NOMA as it performs better over
		conventional MA by selecting users whose channel conditions are more distinctive.
	The SIC performed in conventional NOMA enables users with better channel conditions to decode not only their own information but also the messages intended for other users. This additional information has been exploited in \cite{MAIN} to propose a cooperative NOMA (CO-NOMA) protocol. The users who decode other users' messages, become relays and transmit the other users' messages based on the NOMA principle in the cooperation phase, yielding a higher signal-to-noise ratio (SNR) at the weak users. CO-NOMA was shown to yield a lower mutual outage probability (MOP) and as a result, has better reliability compared to NOMA. The MOP metric, initially proposed in \cite{BibKamel1,BibKamel2}, is very useful  for performance evaluation of a multi-user network. A lower MOP implies a lower outage probability for all users \cite{BibKamel1,BibKamel2}. In addition, in \cite{MAIN} it was shown that each user in CO-NOMA is able to achieve full diversity.

	There have been several works discussing power allocation for NOMA, see \cite{Maraqa19} and the references therein, where various rate optimal NOMA schemes are discussed. However, there has been limited work related to the impact of power allocation on performance of CO-NOMA \cite{BibChen2007,BibLiu2018}. Power allocation to maximize the minimum achievable rate for a two-user  half-duplex CO-NOMA (HDCO-NOMA) and for a two-user full-duplex CO-NOMA (FDCO-NOMA) was proposed in \cite{BibChen2007, BibLiu2018} and it was shown that adapting the transmit power in the cooperative phase can improve the user throughput. Moreover, a hybrid of HDCO-NOMA and FDCO-NOMA was also proposed to further enhance the user throughput \cite{BibChen2007, BibLiu2018}. The spectral efficiency of CO-NOMA is enhanced by employing full duplex transmission which reduces the time slots in the cooperation phase \cite{BibLiu2018, BibWang2019}. In addition, the energy efficiency of CO-NOMA systems can be improved by energy harvesting, see \cite{BibWang2019} and references therein. In order to improve the data-rate of broadcasting and multi-casting in vehicular communication networks using CO-NOMA, power allocation to maximize the minimum throughput for both HDCO-NOMA and FDCO-NOMA was proposed in \cite{BibLiu2019}.  	
	
	{	
		In \cite{MAIN, BibChen2007, BibLiu2018, BibWang2019, BibLiu2019, BibQasim}, the power coefficients and the decoding sequence used in CO-NOMA was based on the NOMA principle in which the user with weak channel (the weak user) is allocated higher power and the user with the strong channel (the strong user) is allocated a lower power. In NOMA, all users receive only a single copy of the signal. This power allocation does reduce MOP of the network in case of NOMA. However, in CO-NOMA, the weak users receive more copies of the signal during the cooperative phase, unlike NOMA. Therefore, for CO-NOMA, a power allocation borrowed from NOMA is sub-optimal as it does not consider the signal copies received during the cooperative phase. If the power allocation from NOMA is utilized in CO-NOMA, the weak users receive more copies of the signal in addition to being allocated the higher power. This does reduce the outage probability of the weak user but on the other hand it does not reduce the outage probability of the strong user which only receives a single copy of the signal with lowest power. This results in a larger MOP implying a higher probability that any of the users is in outage in the network.

		In this letter, a heuristic power allocation scheme is proposed for CO-NOMA which reduces the MOP. In the proposed scheme, different from conventional CO-NOMA, higher power is assigned to the strong user. Therefore, this user only receives a single copy of the signal but benefits from the higher power. The weak user benefits by receiving multiple copies of the signal which provides diversity gain and improves reliability. Numerical simulations show that the proposed power allocation, has a lower MOP and improves reliability compared to the conventional CO-NOMA scheme considered in \cite{MAIN, BibChen2007, BibLiu2018, BibWang2019, BibLiu2019,BibQasim}. In addition, for a two-user cooperative network, exact closed-form expressions for MOP and user outage probability are derived and it is shown that each user achieves full diversity. Moreover, in a practical non-power-balanced scenario with 90\% non-outage, simulation results show that the proposed allocation can achieve approximately 30\% higher rate at 15 dB, compared to conventional CO-NOMA. \comment{ Exact outage probability expressions for the strongest user and the weakest user are also presented. }
	}

	\section{System Model and the Proposed Scheme}
	
	Consider a downlink communication scenario in a network having one base station (BS) and $K$ users. Signal transmission from the BS to the users is carried out in two phases, 1) Direct phase and 2) Cooperative phase \cite{MAIN}. In the direct phase, the BS broadcasts a combination of messages of $K$ users, $\sum_{m=1}^K p_m s_m$, where $ {s_m}$ is the message signal of the $m$-th user and $p_m$ is the power allocation coefficient for the $m$-th user. The signal received at the $k$-th user, in the direct phase (first time slot), is given by

	\begin{equation*} 
	y_{1,k} = \sqrt{\rho}{h_k} \sum_{m=1}^K p_m s_m + n_{1,k},
	\end{equation*} 
	where ${\rho}$ is the transmit power, ${n_{1,k}}$ is the noise added at the $k$-th user and ${h_k}$ is the channel gain between the BS and the $k$-th user. The channel is assumed to be independent and identically distributed (i.i.d.) and has Rayleigh fading. Therefore, $h_k$ is a zero mean complex Gaussian random variable (RV). Without the loss of generality, $h_k$ is assumed to have unit variance. As a result, the channel power gain, $|h_k|^2$ is a unit mean exponential RV\footnote{The mean power of the channel is assumed to be fixed over both the phases. However, the channel gains may vary after each time slot.} \cite{6649263, 8007301, 6478151,6553287}. {The noise at the users is assumed to be additive white Gaussian noise with zero mean and unit variance.} During the direct phase, the signal at the users is decoded using the principle of SIC. The users are able to decode the stronger signal, subtract it from the combined signal, and then extract the intended messages of all the users from the residue.
	
	In the cooperative phase, other users' messages which are decoded using SIC, are re-transmitted to achieve better reliability and improved quality-of-service (QoS) at each user. Similar to \cite{MAIN}, it is assumed that each user re-transmits with power $\rho$. Users are ordered based on their channel power as $ |h_1|^2 \leq {|h_2|^2} \leq .....\leq |h_K|^2$. In the first time slot of this phase, the $K$-th user (which is the user with the best channel) broadcasts, to the other users, a signal which is a combination of the other ${(K-1)}$ user's messages ${{\sum_{m=1}^{K-1}\ q_{K,m} s_m }}$, where $q_{K,m}$ denote the power allocation coefficient for the $m$-th user at the $K$-th user and  ${ {\sum_{m=1}^{K-1}\ {q_{K,m}^2} = 1}}$ \cite{MAIN}. The signal received at the $k$-th user from the $K$-th user in this time slot is
	\begin{equation*}
	{y_{2,k} = \sqrt{\rho}{g_{K,k}\sum_{m=1}^{K-1} q_{K,m} s_m + n_{2,k}}},
	\end{equation*}
	where  $k \in \left\{1,2,\cdots, (K-1)\right\}$, $n_{2,k}$ is the noise added at the $k$-th user, ${g_{i,j}}$ denotes the zero mean complex Gaussian channel gain from the $i$-th user to the $j$-th user. The ${(K-1)}$-th user combines the signals from both phases by using maximum ratio combining~(MRC) before decoding each users data. This user then broadcasts, to the remaining users, a signal which is a combination of the other ${(K-2)}$ users' messages. During the ${n}$-th time slot of this phase, ${1\leq n \leq (K-1)}$, the ${(K-n+1)}$-th user broadcasts the combination of ${(K-n)}$ messages with the power allocation coefficients ${ q_{K-n+1,m}}$ and the signal received at the ${k}$-th user, where $k \in \left\{1, 2,\cdots,(K-n)\right\}$, is
\vspace{-6pt}
	\begin{equation*}
	{y_{n+1,k} = {\sqrt{\rho}} {g_{K-n+1,k} \sum_{m=1}^{K-n} q_{K-n+1,m} s_m + n_{n+1,k}}},
	\end{equation*}
    where $n_{n+1,k}$ is the noise added at the $k$-th user.\\
	\emph{\bf Proposed Power Allocation:} 
	Users are ordered based on their channel power as $ {|h_1|^2 } \leq {|h_2|^2} \leq .....\leq |h_K|^2$.
	\comment{
		\begin{equation}
		{|h_1|^2 } \leq {|h_2|^2} \leq .....\leq |h_K|^2 \label{A}
		\end{equation}
	}
	The power coefficients allocated to the ordered users during the first phase satisfy

	\begin{equation}
	\label{eq:powerallocation}
	{p_1^2 } \leq {p_2^2} \leq .....\leq p_K^2,
	\end{equation}
	with ${ {\sum_{m=1}^K\ {p_m}^2 = 1}}$. This new power allocation is different from the one proposed in \cite{MAIN} where the power allocation to the ordered users is ${p_K^2 \leq {p_{K-1}^2} \leq .....\leq {p_1^2 }}$, which allocates high power to the weak user. The weak user also receives multiple copies of the signal, whereas the strong user is allocated lower power and only receives a single copy of the signal. This allocation in \cite{MAIN} benefits the users with the poor channel and neglects the users with the stronger channel.
	
	On the contrary, the new power allocation in \eqref{eq:powerallocation} allocates higher power to the strong user and lower power to the weak user. The weak user benefits from diversity as it receives multiple copies of the signal. Moreover, as the power allocation in \eqref{eq:powerallocation} is opposite to that in \cite{MAIN}, the sequence of decoding using SIC is also reversed resulting in different signal-to-interference plus noise ratio (SINR) at the respective users.
	
	{
		
		\emph{\bf Proposed Decoding Sequence:}
		In the proposed scheme, as the strongest user is allocated the maximum power, therefore, each user has to decode the data of stronger users before decoding their own data. Assuming, that there are $K$ users, the steps for decoding are as follows:
		\noindent
		\emph{First time slot:} 
		\begin{enumerate}
			\item The $K$-th user (which is the strongest user) will decode its data first as it is allocated the highest power. On the contrary, in conventional CO-NOMA, the $K$-th user data is decoded at the end as it is allocated the least power. 
			\item The $K$-th user will perform SIC and decode all other user's data.
			\item All the remaining users will decode only the $K$-th user's data and perform SIC. They will neither decode their own data nor any other users data.
		\end{enumerate}
		\noindent
		\emph{Cooperative Phase - $n$-th time slot ($1<n<K$):} 
		\begin{enumerate}
			\item The $(K-n+1)$-th user will receive signal from the $(K-n+2)$-th user where the $(K-n+1)$-th user data is allocated the highest power. It will combine the signal copy received in this time slot and all the previous time slots using MRC and then decode its data. 
			\item The $(K-n+1)$-th user will perform SIC and decode the data of all the remaining users based on the combined signal from existing and all previous time slots. 
			\item All the remaining users will decode only the $(K-n+1)$-th user data and perform SIC. They will neither decode their own data nor any other users data.
		\end{enumerate}

		\noindent
		\emph{$K$-th time slot:} 
		\begin{enumerate}
			\item The weakest user will receive signal from user 2. It will combine the signal copy received in this time slot and all the previous time slots using MRC and then decode its data.
		\end{enumerate}			
		
		It can be noted that, similar to conventional CO-NOMA, the weakest user data is decoded in the $K$-th time slot. No additional time slot is needed for decoding the weak user data compared to orthogonal multiple access (OMA) or conventional CO-NOMA. The decoding steps for the proposed scheme, for $K=3$, are summarized in Table 1. 
	}
	
	%

	\begin{table*}
	{
	\begin{sideways}	
		\begin{minipage}{20.5cm}
  
				\caption{Signal Decoding and SINR for a three user network where $|h_{1}|^{2}<|h_{2}|^{2}<|h_{3}|^{2}$ and $|p_{1}|^{2}<|p_{2}|^{2}<|p_{3}|^{2}$}
\hspace{-1.7cm}
					\begin{tabular}{llll}
						\toprule
						\textbf{\textsc{User}} & \textbf{\textsc{ Time Slot 1}} & \textbf{\textsc{Time Slot 2}}& \textbf{\textsc{Time Slot 3}}\\
						\toprule
						Strong
						& \begin{tabular}{@{}l@{}}
							$y_{1,3}=h_{3}\sqrt{\rho}\left(p_{1}s_{1}+p_{2}s_{2}+p_{3}s_{3}\right)+n_{1,3}$  \\ 
							1. Decode $s_3$:~$\gamma_{1,3,strong}=\frac{p_{3}^{2}|h_{3}|^{2}}{\left(p_{1}^{2}+p_{2}^{2}\right)|h_{3}|^{2}+\frac{1}{\rho}}=\gamma_{3}$ \\
							2. Decode $s_2$:~$\frac{p_{2}^{2}|h_{3}|^{2}}{p_{1}^{2}|h_{3}|^{2}+\frac{1}{\rho}}$\\
							3. Decode $s_1$:~$\gamma_{1,3,weak}=\rho p_{1}^{2}|h_{3}|^{2}$\end{tabular}
						& 1. Broadcast $\sqrt{\rho} \left(q_{3,1}s_1+q_{3,2}s_2\right)$ to weaker users
						& \begin{tabular}{@{}l@{}} 
							
						\end{tabular}\\
						\midrule
						
						Mid 
						& \begin{tabular}{@{}l@{}}
							$y_{1,2}=h_{2}\sqrt{\rho}\left(p_{1}s_{1}+p_{2}s_{2}+p_{3}s_{3}\right)+n_{1,2}$\\
							1. Decode $s_3$:~$\gamma_{1,2,strong}=\frac{p_{3}^{2}|h_{2}|^{2}}{\left(p_{1}^{2}+p_{2}^{2}\right)|h_{2}|^{2}+\frac{1}{\rho}}$
							\\
							
							2. Perform SIC; $s_2$ and $s_1$ are not decoded in this time slot.

						\end{tabular} 
						& \begin{tabular}{@{}l@{}} $	y_{2,2}=g_{3,2}\sqrt{\rho}\left(q_{3,1}s_{1}+q_{3,2}s_{2}\right)+n_{2,2}$ \\
							
							1. Combine the signal $s_2$ from both time slots using MRC and decode 
							\\
							\hspace{10pt} SINR for $s_2$:~$\gamma_{2}=\frac{p_{2}^{2}|h_{2}|^{2}}{p_{1}^{2}|h_{2}|^{2}+\frac{1}{\rho}}+\frac{q_{3,2}^{2}|g_{3,2}|^{2}}{q_{3,1}^{2}|g_{3,2}|^{2}+\frac{1}{\rho}}$ 
							
							\\
							2. Perform SIC and decode $s_1$:
							\\ 
							
							\hspace{10pt} SINR for $s_1$:~$\gamma_{2,2,weak}=\rho p_{1}^2|h_2|^{2}+\rho q_{3,1}^2|g_{2,1}|^{2}$\\

						\end{tabular}
						& \begin{tabular}{@{}l@{}} 
							
							1. Broadcast $\sqrt{\rho} q_{2,1}s_1$ to weakest user, \\
where $q_{2,1}=1$.
							
						\end{tabular}\\
						\midrule
						Weak 
						& \begin{tabular}{@{}l@{}}
							
							$y_{1,1}=h_{1}\sqrt{\rho}\left(p_{1}s_{1}+p_{2}s_{2}+p_{3}s_{3}\right)+n_{1,1}$
							\\
							1. Decode $s_3$:~$\gamma_{1,1,strong}=\frac{p_{3}^{2}|h_{1}|^{2}}{\left(p_{1}^{2}+p_{2}^{2}\right)|h_{1}|^{2}+\frac{1}{\rho}}$
							\\
							
							2. Perform SIC; $s_2$ and $s_1$ are not decoded in this time slot.

						\end{tabular} 
						& \begin{tabular}{@{}l@{}} 
							$y_{2,1}=g_{3,1}\sqrt{\rho}\left(q_{3,1}s_{1}+q_{3,2}s_{2}\right)+n_{2,1}$
							\\	
							1. Combine the signal $s_2$ from both time slots using MRC and decode 
							\\
							\hspace{10pt} SINR for  $s_2$:~$\gamma_{2,1,mid}=\frac{p_{2}^{2}|h_{1}|^{2}}{p_{1}^{2}|h_{1}|^{2}+\frac{1}{\rho}}+\frac{q_{3,2}^{2}|g_{3,1}|^{2}}{q_{3,1}^{2}|g_{3,1}|^{2}+\frac{1}{\rho}}$
							
							\\
							2. Perform SIC, but do not decode $s_1$	
						\end{tabular}
						& \begin{tabular}{@{}l@{}} 
							$y_{3,1}=g_{2,1}\sqrt{\rho}q_{2,1}s_{1}+n_{3,1}$ \\
							1. Combine the signal $s_1$, from all time\\
							 slots  using MRC and decode:\\
							
							\hspace{10pt} SINR for $s_1$:~$\gamma_{1}=\rho p_{1}^{2}|h_{1}|^{2}$\\
\hspace{1cm}$+ \rho q_{3,1}^2|g_{3,1}|^2+\rho |g_{2,1}|^{2}$
							
						\end{tabular}\\
						\midrule 
						\bottomrule 
				\end{tabular}
		
		\end{minipage}
\end{sideways}
	}
	\end{table*}

	\section{Performance Analysis}
	In the cooperative phase, the received signals at the users are combined using MRC. The combined SINR for the ${(K-n)}$-th
	user, for ${0\leq n<K-1}$, is given as
	\begin{equation}
	\label{eq:SINR1}
	\gamma_{K-n} = {\frac{|h_{K-n}|^2 p_{K-n}^2}{|h_{K-n}|^{2}\sum_{m={1}}^{K-n-1}p_m^2+\frac{1}{\rho}}}
		+ {\sum_{i={1}}^{n}\frac{|g_{K-i+1,K-n}|^2q_{K-i+1,K-n}^2}{|g_{K-i+1,K-n}|^{2}\sum_{m={1}}^{K-n-1}q_{K-i+1,m}^2+\frac{1}{\rho}}}.
	\end{equation}
	It can be noted that the opposite sequence of SIC in the proposed scheme makes~\eqref{eq:SINR1} different from the traditional CO-NOMA in terms of power coefficients as well as the limits of summation~(cf. \cite[Eq. (7)]{MAIN} and~\cite[Eq. (3)]{BibLiu2019}) . For ${n=K-1}$, i.e. the user with the worst channel, the SINR is given as
	\begin{equation}
	\label{eq:SINR2}
	\gamma_{1}= \rho |h_1|^2 p_1^2 + \rho\sum_{i={1}}^{K-1} |g_{K-i+1,1}|^2 q_{K-i+1,1}^2.
	\end{equation}
	Again it can be noted that the SINR of the weak user in \eqref{eq:SINR2} is different compared to the scheme proposed in \cite[Eq. (7)]{MAIN}.
	
	The performance metrics considered to analyze the proposed scheme are outage probability of the respective users and the MOP, which is defined as
	\begin{equation}
	\label{eq:MOPDef}
	\mathcal{M_O} = 1- \text{Pr}\left(\gamma_1>\phi_1, \gamma_2>\phi_2, \ldots , \gamma_K>\phi_K  \right),
	\end{equation}
	where $\phi_k=2^{R_k}-1$, $R_k \in \{1,2,3...\}$ and denotes the desired rate for communication of the $k$-th user in bits per channel use \cite{MAIN}. The outage probability of the respective users will always be less than or equal to the MOP (proof provided in Appendix B). So, if a network is designed to achieve a MOP equal to $\alpha$, it will be guaranteed that each user in the network will achieve an outage probability less than or equal to $\alpha$ and experience a good QoS.

\noindent
	\emph{\bf Outage Probability of the Weakest User:}
	
	The SNR of the weakest user, in \eqref{eq:SINR2}, can be expressed as $\gamma_{1}=\sum_{i=0}^{K-1}X_{i}$ where $X_{0}=\rho|h_{1}|^{2}p_{1}^{2}$ and $X_{i}=\rho q_{K-i+1,1}^{2}|g_{K-i+1,1}|^{2}$. Denoting an exponential RV, $A$, with mean $\lambda$ as $A\sim\exp\left(\lambda\right)$, it can be noted that $X_{0}\sim\exp\left(\frac{\rho p_{1}^{2}}{K}\right)$ ($|h_{1}|^{2}$ is the minimum of $K$ exponential RVs) and $X_{i}\sim\exp\left(\rho q_{K-i+1,1}^{2}\right)$. Thus, $\gamma_{1}$ is a sum of independent and non-identical (i.n.i.d) exponential RVs having CDF 
	\begin{equation}
	\label{eq:CDFWeakUser}
	F_{\gamma_{1}}\left(\gamma\right)=\sum_{i=0}^{K-1}C_{i}\left(1-e^{-\frac{\gamma}{\lambda_{i}}}\right),
	\end{equation}
	where $C_{i}=\prod_{j\neq i}\frac{\frac{1}{\lambda_{j}}}{\frac{1}{\lambda_{j}}-\frac{1}{\lambda_{i}}}$, $\lambda_{0}=\frac{\rho p_{1}^{2}}{K}$ and $\lambda_{i}=\rho q_{K-i+1,1}^{2}$ \cite{BibRossBook}. \emph{ The outage probability of the weakest user can be obtained by evaluating the CDF in \eqref{eq:CDFWeakUser} at $\gamma=\phi_1$.}
	
	\noindent
	\emph{\bf Outage Probability of the Strongest User:}
	
	The SINR of the remaining users, in \eqref{eq:SINR1}, can be expressed as $\gamma_{K-n}=Y_{n}+\sum_{i=1}^{n}Z_{i}$, where $Y_{n}=\frac{|h_{K-n}|^{2}p_{K-n}^{2}}{|h_{K-n}|^{2}\sum_{m=1}^{K-n-1}p_{m}^{2}+\frac{1}{\rho}}$ and $Z_{i}=\frac{|g_{K-i+1,K-n}|^{2}q_{K-i+1,K-n}^{2}}{|g_{K-i+1,K-n}|^{2}\sum_{m=1}^{K-n-1}q_{K-i+1,m}^{2}+\frac{1}{\rho}}$.
	
The CDF of $Y_{n}$, $F_{Y_{n}}(y)=\text{Pr}\left( \frac{|h_{K-n}|^{2}p_{K-n}^{2}}{|h_{K-n}|^{2}\sum_{m=1}^{K-n-1}p_{m}^{2}+\frac{1}{\rho}} \leq y  \right)$, is simplified as
	\begin{equation}
	F_{Y_{n}}(y)=\begin{cases}
	F_{|h_{K-n}|^{2}}\left(\frac{1}{\rho}\omega_{0,n}(y)\right) & ;\hfill y\leq\zeta_{0,n}\\
	\qquad\hfill1 & ;\hfill y>\zeta_{0,n},
	\end{cases}
	\end{equation}
	where $\omega_{0,n}(y)=\frac{y}{\left(p_{K-n}^{2}-y\sum_{m=1}^{K-n-1}p_{m}^{2}\right)\,}$,
	$\zeta_{0,n}=\frac{p_{K-n}^{2}}{\sum_{m=1}^{K-n-1}p_{m}^{2}}$ and
	$|h_{K-n}|^2$ is the $\left(K-n\right)$-th highest channel power gain among $K$ channel gains. Using order statistics, for the Rayleigh fading scenario, the CDF is obtained as
	\comment{
		channel gains and has distribution 
		\begin{equation}
		F_{H_{m}}(h)=\sum_{i=0}^{m-1}\binom{M}{i}\left(F_{H_{i}}\left(h\right)\right)^{M-i}\left(1-F_{H_{i}}\left(h\right)\right)^{i}
		\end{equation}

		\begin{equation}
		F_{h_{n+1}}(h)=\sum_{i=0}^{n}\binom{K}{i}\left(F_{h_{i}}\left(h\right)\right)^{K-i}\left(1-F_{h_{i}}\left(h\right)\right)^{i}
		\end{equation}
	}
	\begin{equation}
	\label{eq:CDFYn}
	F_{Y_{n}}(y)=\begin{cases}
	\sum_{i=0}^{n}\binom{K}{i}\left(1-e^{-\frac{1}{\rho}\omega_{0,n}(y)}\right)^{K-i}\left(e^{-\frac{1}{\rho}\omega_{0,n}(y)}\right)^{i} & ;\hfill y\leq\zeta_{0,n}\\
	\qquad\hfill1 & ;\hfill y>\zeta_{0,n}.
	\end{cases}
	\end{equation}
	Similar to $Y_n$, the CDF of $Z_{i}$, $F_{Z_{i}}(z)=\text{Pr}\left( \frac{|g_{K-i+1,K-n}|^{2}q_{K-i+1,K-n}^{2}}{|g_{K-i+1,K-n}|^{2}\sum_{m=1}^{K-n-1}q_{K-i+1,m}^{2}+\frac{1}{\rho}} \leq z \right)$, yields
	\begin{equation}
	F_{Z_{i}}(z)=\begin{cases}
	F_{|g_{K-i+1,K-n}|^{2}}\left(\frac{1}{\rho}\omega_{i,n}\left(z\right)\right) & ;\hfill z\leq\zeta_{i,n}\\
	\qquad\hfill1 & ;\hfill z>\zeta_{i,n},
	\end{cases}
	\end{equation}
	where $\omega_{i,n}\left(z\right)=\frac{z}{\left(q_{K-i+1,K-n}^{2}-z\sum_{m=1}^{K-n-1}q_{K-i+1,m}^{2}\right)\,}$,
	$\zeta_{i,n}=\frac{q_{K-i+1,K-n}^{2}}{\sum_{m=1}^{K-n-1}q_{K-i+1,m}^{2}}$ and $|g_{K-i+1,K-n}|^2$ are unit mean i.i.d. exponential RVs, which implies that
	\begin{equation}
	\label{eq:CDFZn}
	F_{Z_{i}}(z)=\begin{cases}
	\left(1-e^{-\frac{1}{\rho}\omega_{i,n}\left(z\right)}\right) & ;\hfill z\leq\zeta_{i,n}\\
	\qquad\hfill1 & ;\hfill z>\zeta_{i,n}.
	\end{cases}
	\end{equation}
	\emph{The outage probability of the strongest user can be obtained by substituting $n=0$ and $\gamma=\phi_K$ in the CDF in \eqref{eq:CDFYn}.} 
	
	Derivation of the analytical expression for the outage probability for the remaining users leads to an intractable analysis. Specifically, the moment generating function (MGF) of $Y_n$ and $Z_i$ cannot be obtained in closed-form. In order to gain some insight into the outage performance of the users, the diversity order achieved by each user is calculated.

	\noindent
	\emph{\bf Diversity Order:}
	For two non-negative RVs, $A$ and $B$, it can be shown
	\begin{equation*}
	\begin{split}
	&\text{Pr}(A+B\leq x)=\int\int_{a+b\leq x}f_{A,B}(a,b)dadb
	\\
	&\leq\int\int_{a\leq x,b\leq x}f_{A,B}(a,b)dadb=\text{Pr}(A\leq x,B\leq x).
	\end{split}
	\end{equation*}
	This result can be generalized for a large number of non-negative RVs, $A_{j}$, $j\in\left\{ 1,..,\infty\right\} $ as
	\begin{equation}
	\label{eq:DivIdentity}
	\text{Pr}\left(\sum_{i}A_{i}\leq x\right)\leq \text{Pr}\left(\cap_{i}\left(A_{i}\leq x\right)\right).
	\end{equation}

	In order to obtain the diversity order for the weakest user, utilizing the identity in \eqref{eq:DivIdentity}, the CDF in \eqref{eq:CDFWeakUser} is bounded as
	\begin{equation}
	F_{\gamma_{1}}\left(\gamma\right)=\text{Pr}\left(\sum_{i=0}^{K-1}X_{i}\leq\gamma\right)\leq\prod_{i=0}^{K-1}\text{Pr}\left(X_{i}\leq\gamma\right).
	\end{equation}
	As $X_i$ are independent, the bounded CDF is expressed as 	
	\begin{equation}
	F_{\gamma_{1}}\left(\gamma\right)\leq\prod_{i=0}^{K-1}\left(1-e^{-\frac{\gamma}{\lambda_{i}}}\right).
	\end{equation}
	As $\rho\rightarrow\infty$, implies $\frac{\gamma}{\lambda_{i}} \rightarrow 0$. As argument of exponential function approaches 0, it can be approximated using Taylor series as $e^x\approx1+x$ which implies $1-e^{-\frac{\gamma}{\lambda_{i}}} \approx 1-\left(1-\frac{\gamma}{\lambda_{i}}\right)$, yielding
	\begin{equation}
	\label{eq:AsymptoticWeakUser}
	F_{\gamma_{1}}\left(\gamma\right)\leq F_{\gamma_{1}}^{\infty}\left(\gamma\right)=\prod_{i=0}^{K-1}\left(\frac{\gamma}{\lambda_{i}}\right)=\frac{\gamma K}{\rho p_{1}^{2}}\prod_{i=1}^{K-1}\left(\frac{\gamma}{\rho q_{K-i+1,1}^{2}}\right).
	\end{equation}
	The diversity order for the weakest user can be obtained as\\ $\mathcal{D}_{1}=-\lim_{\rho\rightarrow\infty}\frac{\log\left(F_{\gamma_{1}}^{\infty}\left(\gamma\right)\right)}{\log\left(\rho\right)}=-\lim_{\rho\rightarrow\infty}\frac{K\log\left(\frac{1}{\rho}\right)+K\log\left(\frac{\gamma K}{ p_{1}^{2}}\prod_{i=1}^{K-1}\left(\frac{\gamma}{q_{K-i+1,1}^{2}}\right)\right)}{\log\left(\rho\right)}=K$.

	In order to obtain the diversity order for the remaining users, $0\leq n <K-1$, again utilizing the identity in \eqref{eq:DivIdentity}, the CDF of $\gamma_{K-n}$ can be bounded as
	\begin{equation}
	\label{eq:DivCDFKn}
	F_{\gamma_{K-n}}\left(\gamma\right)=\text{Pr}\left(Y_{n}+\sum_{i=1}^{n}Z_{i}\leq\gamma\right)\leq F_{Y_{n}}\left(\gamma\right)\prod_{i=1}^{n}F_{Z_{i}}\left(\gamma\right).
	\end{equation}
	Substituting the CDFs from \eqref{eq:CDFYn}	and \eqref{eq:CDFZn} into \eqref{eq:DivCDFKn} yields	
	\begin{equation*}
	\begin{split}
 F_{\gamma_{K-n}}\left(\gamma\right)  \leq
 \sum_{i=0}^{n}\binom{K}{i}\left(1-e^{-\frac{1}{\rho}\omega_{0,0}(\gamma)}\right)^{K-i}\left(e^{-\frac{1}{\rho}\omega_{0,0}(\gamma)}\right)^{i} \prod_{k=1}^{n}\left(1-e^{-\frac{1}{\rho}\omega_{k,n}\left(\gamma\right)}\right).
	\end{split}
	\end{equation*}
	Similar to \eqref{eq:AsymptoticWeakUser}, as $\rho\rightarrow\infty$, argument of exponential function approaches 0 and can be approximated using Taylor series as
	\comment{	
		\[
		F_{\gamma_{K-n}}^{\infty}\left(\gamma\right)=\sum_{i=0}^{n}\binom{K}{i}\left(\frac{1}{\rho}\omega_{0,0}(\gamma)\right)^{K-i}\left(1-\frac{1}{\rho}\omega_{0,0}(\gamma)\right)^{i}\prod_{i=1}^{n}\left(\frac{1}{\rho}\omega_{i,n}\left(\gamma\right)\right)
		\]
		
	}	
	\begin{equation}
	\label{eq:AsymptoticRemainingUser}
	F_{\gamma_{K-n}}\left(\gamma\right)\leq F_{\gamma_{K-n}}^{\infty}\left(\gamma\right)=\binom{K}{n}\left(\frac{1}{\rho}\omega_{0,0}(\gamma)\right)^{K-n}\left(\frac{1}{\rho}\right)^{n}\prod_{i=1}^{n}\left(\omega_{i,n}\left(\gamma\right)\right).
	\end{equation}
	\comment{	
		
		\[
		F_{\gamma_{K-n}}^{\infty}\left(\gamma\right)=\left(\frac{1}{\rho}\right)^{K}\binom{K}{n}\left(\omega_{0,0}(\gamma)\right)^{K-n}\prod_{i=1}^{n}\left(\omega_{i,n}\left(\gamma\right)\right)
		\]
		
	} 	
	The diversity order for the $\left(K-n\right)$-th user can be obtained as $\mathcal{D}_{K-n}=-\lim_{\rho\rightarrow\infty}\frac{\log\left(F_{\gamma_{K-n}}^{\infty}\left(\gamma\right)\right)}{\log\left(\rho\right)}=-\lim_{\rho\rightarrow\infty}\frac{K\log\left(\frac{1}{\rho}\right)+K\log\left(\binom{K}{n}\omega_{0,0}^{K-n}(\gamma)\prod_{i=1}^{n}\left(\omega_{i,n}\left(\gamma\right)\right)\right)}{\log\left(\rho\right)}=K$. 
	
	\emph{All users in the proposed scheme achieve full diversity.}
	
	\noindent
	\emph{\bf Mutual Outage Probability}	
	
	As discussed previously, deriving the analytical expression for the outage probability for the $(K-n)$-th user, where $0<n<K-1$, leads to an intractable analysis. Similarly, derivation of the MOP expression, for the $K>2$ case, leads to an intractable analysis. For the $K=2$ user case, the exact closed-form MOP expression can be derived and is given in Proposition 1. For a network with two users, using \eqref{eq:SINR1}, the SINR at strong user is given as
	\begin{equation}
	\label{eq:SINRuser2}
	\gamma_{2}=\frac{\rho p_{h}|h_{2}|^{2}}{\rho p_{l}|h_{2}|^{2}+1},
	\end{equation}
	where the higher power is denoted as $p_{h}=p_{2}^{2}$ and the lower power value is denoted as $p_{l}=p_{1}^{2}$. Using \eqref{eq:SINR2}, the SNR at the weak user is given as
	\begin{equation}
	\label{eq:SINRuser1}
	\gamma_{1}=\rho p_{l}|h_{1}|^{2}+\rho|g_{2,1}|^{2}.
	\end{equation}
	Again it can be noted that the expression for SINR in \eqref{eq:SINRuser2} and SNR in \eqref{eq:SINRuser1} are different from previous works in \cite{MAIN},\cite[Eq. (2)]{BibLiu2018} and thus, yield a different statistical analysis. 
	
	\begin{Proposition}
		The MOP for the two user CO-NOMA system using the power allocation in \eqref{eq:powerallocation} is given in \eqref{eq:MOPFinal} where
		$\phi_2<\frac{p_h}{p_l}$, $\scriptstyle\beta=\frac{1}{\rho p_{h}}\frac{\phi_{2}}{\left(1-\phi_{2}\frac{p_{l}}{p_{h}}\right)}$, $\scriptstyle\Delta=\min\left\{c,\frac{\phi_{1}}{\rho}\right\} $ and
		$\scriptstyle c=\frac{1}{\rho}\left(\phi_{1}-\frac{\phi_{2}}{\left(\frac{p_{h}}{p_{l}}-\phi_{2}\right)}\right)$. If $\phi_2>\frac{p_h}{p_l}$, then at high SNR (large $\rho$), the strong user will be unable to meet the rate requirement and $\mathcal{M_O}=1$. 
		\begin{figure*}[t]
			\begin{equation} 
			\label{eq:MOPFinal}
			\scriptsize
			\begin{split}
			\mathcal{M_O}=1
			&-e^{-\frac{\phi_{1}}{\rho}}\left(1-\left(1-e^{-\frac{1}{\rho
					p_{h}}\frac{\phi_{2}}{\left(1-\phi_{2}\frac{p_{l}}{p_{h}}\right)}}\right)^{2}\right)
			\\
			&-\left(\left(\frac{e^{-\frac{2}{p_{l}}\frac{\phi_{1}}{\rho}}}{\left(1-\frac{2}{p_{l}}\right)}\left(1-e^{-\left(1-\frac{2}{p_{l}}\right)\Delta}\right)\right)+\left(\frac{2e^{-\beta}e^{-\frac{1}{p_{l}}\frac{\phi_{1}}{\rho}}}{\left(1-\frac{1}{p_{l}}\right)}\left(e^{-\frac{1}{\rho}\left(\phi_{1}-\frac{\phi_{2}\left(p_{l}-1\right)}{p_{h}\left(1-\frac{p_{l}}{p_{h}}\phi_{2}\right)}\right)}-e^{-\frac{\phi_{1}}{\rho}\left(1-\frac{1}{p_{l}}\right)}\right)-e^{-2\beta}\left(e^{-c}-e^{-\frac{\phi_{1}}{\rho}}\right)\right)\right)
			\end{split}
			\end{equation}
			\rule{7.1in}{.1pt}
			\vspace{-16pt}
		\end{figure*}
	\end{Proposition}
	\begin{proof}
		See Appendix A.
	\end{proof}
	
	\noindent	\emph{Remark:} The MOP in \eqref{eq:MOPFinal} is obtained in closed-form and is given in terms of exponential functions which can be easily evaluated using existing mathematical packages such as MATLAB\textsuperscript{\textregistered}.
	The derived expression is useful for the practical scenario, where a 2-user CO-NOMA system may be deployed\footnote{For a large number of users, the deployment of NOMA and CO-NOMA becomes extremely challenging, e.g. the decoding at the users becomes more involved, the probability of error propagation increases, resource allocation becomes challenging \cite{BibMakki2020,BibIslam2017}. In order to counter these challenges, user pairing is proposed and most existing works consider only 2-user pairing \cite{BibMakki2020,MAIN, BibIslam2017}.}.
	Moreover, the power allocation which yields the minimum MOP can be obtained quickly using \eqref{eq:MOPFinal} as compared to running time-consuming simulations\footnote{The simulations were carried out on a Intel Core i5 processor @1.10GHz with 8 GB RAM. For the proposed power allocation scheme, the time for calculating the optimal power value using \eqref{eq:MOPFinal}, for a single value of $\rho$ and rate, $R_k$, was approximately 1ms. For conventional CO-NOMA of \cite{MAIN}, the calculation of optimal power value required averaging over $5 \times 10^5$ realizations of the channel and took approximately 180s. The time was calculated using tic/toc function in MATLAB.}.

	\section{Numerical Results}
	Monte-Carlo simulations were performed in MATLAB to analyze and compare the performance of CO-NOMA utilizing the new power allocation in
	\eqref{eq:powerallocation} with the allocation proposed in \cite{MAIN}. CO-NOMA with the new power allocation is denoted as CN-PA and CO-NOMA with the standard power allocation, of \cite{MAIN}, is denoted as CN-SA. The performance is also compared to OMA, where each users' data is transmitted with power $\rho$ in different time slots. The average channel and noise power is assumed to be unity. The simulation curves are generated by averaging over $5 \times 10^5$ realizations of the channel.
	
	\comment{
		\begin{figure}
			\centering
			\includegraphics[width=\linewidth]{images/optimalPowerproposedK2.eps}
			\caption{MOP versus varying the power coefficient, $p_2^2$, and transmit power, $\rho$, for CN-PA, ${K=2}$.}
			\label{fig:optimalPowerproposedK2}
		\end{figure}

	}
	
	\figref{fig:OutageProbK2}~compares the MOP performance, for ${K=2}$, of CN-SA and the proposed CN-PA scheme. For CN-PA, the optimal value of $p_2$ that minimizes \eqref{eq:MOPFinal} and also satisfies \eqref{eq:powerallocation} is calculated numerically in MATLAB. For CN-SA, the optimal power coefficients, which  achieve minimum MOP and satisfy ${p_K^2 \leq {p_{K-1}^2} \leq \cdots \leq {p_1^2 }}$, were calculated using a time-consuming brute-force simulation. 
	The optimal squared power coefficient values for $K=2$ are listed in Table II. It can be observed in \figref{fig:OutageProbK2}, that the proposed CN-PA scheme has a lower MOP compared to the CN-SA scheme. The reason for this is that in CN-SA, the strong user is allocated less power and has a higher outage probability which results in a larger MOP (it is shown in Appendix-B, if any user has a high outage probability, the MOP will be higher than it.). CN-PA allocates higher power to the strong user to reduce its outage probability, as a result of which the overall MOP reduces. Observing the optimal power value for the CN-SA in Table II, $p_1^2=.51$ and as a consequence $p_2^2=.49$. This yields a power-balanced scenario $(p_1^2\approx p_2^2)$, in which the practical SIC decoder might not perform well \cite{BibPan2017}. On the contrary, the difference in the optimal power coefficient values for CN-PA is very large, which is desirable for error free SIC decoding. For comparison, the performance of CN-SA in a non-power balanced scenario, i.e. $p_1^2=0.8$, is also plotted. It can be noted that, in this case, the CN-PA scheme offers a 3 dB gain over the CN-SA scheme when $R_k=1$, $\forall k$. Moreover, it can also be noted that the simulation results match exactly with the derived MOP equation in~\eqref{eq:MOPFinal}.
	
	\comment{\figref{fig:OutageProbK2} also shows the MOP comparison for $K=3$ users, and again in this case, the proposed CN-PA scheme has a lower MOP as compared to the CN-SA scheme.}
	
	\begin{table}
		{
			\caption{Optimal squared power coefficients for $K=2$.} 
			\vspace{-8pt}
			\resizebox{1\columnwidth}{!}
			{	
				\begin{tabular}{c|l|c|c|c|c|c|c|c|c|c}
					& & \textbf{$\rho$} & \textbf{0 dB} & \textbf{3 dB} & \textbf{6 dB}& \textbf{9 dB}& \textbf{12 dB}& \textbf{15 dB}& \textbf{18 dB}& \textbf{21 dB}\\ 
					\hline
					\multirow{2}{*}{$R_k=1$}	&CN-PA&$p_2^2$ & 0.995 &    0.995&    0.845&    0.810&    0.795&    0.790&    0.790&    0.790\\ 
					&CN-SA&$p_1^2$  &0.510&    0.510&    0.510&    0.560&    0.560&    0.560&    0.560&    0.560\\ 
					\hline			
					\multirow{2}{*}{$R_k=2$}	&CN-PA&$p_2^2$ &0.995 &  0.995&    0.995&    0.995&    0.930&    0.895&    0.885&    0.880\\ 
					&CN-SA&$p_1^2$  &0.510&   0.510&    0.510&    0.510&    0.510&    0.510&    0.540&    0.560\\ 
				\end{tabular}
			}
		}
		\vspace{-16pt} 
	\end{table}

	\begin{figure}  
		\centering 
		\includegraphics[width=.8\linewidth]{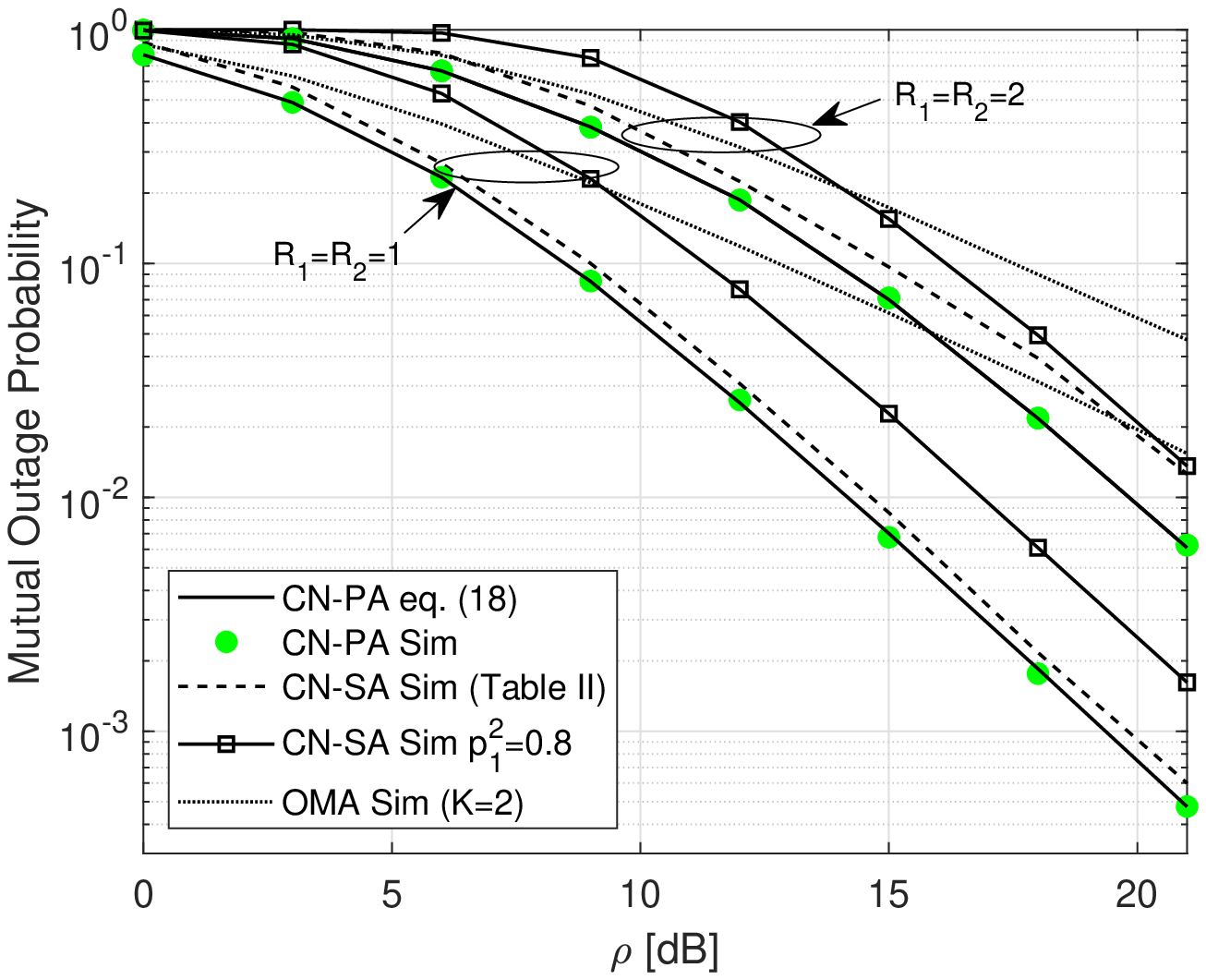}
		\caption{Comparison of MOP of CO-NOMA schemes for $K=2$ users.}
		\label{fig:OutageProbK2}
	\end{figure}

	\begin{figure}
		\centering 
		\includegraphics[width=.8\linewidth]{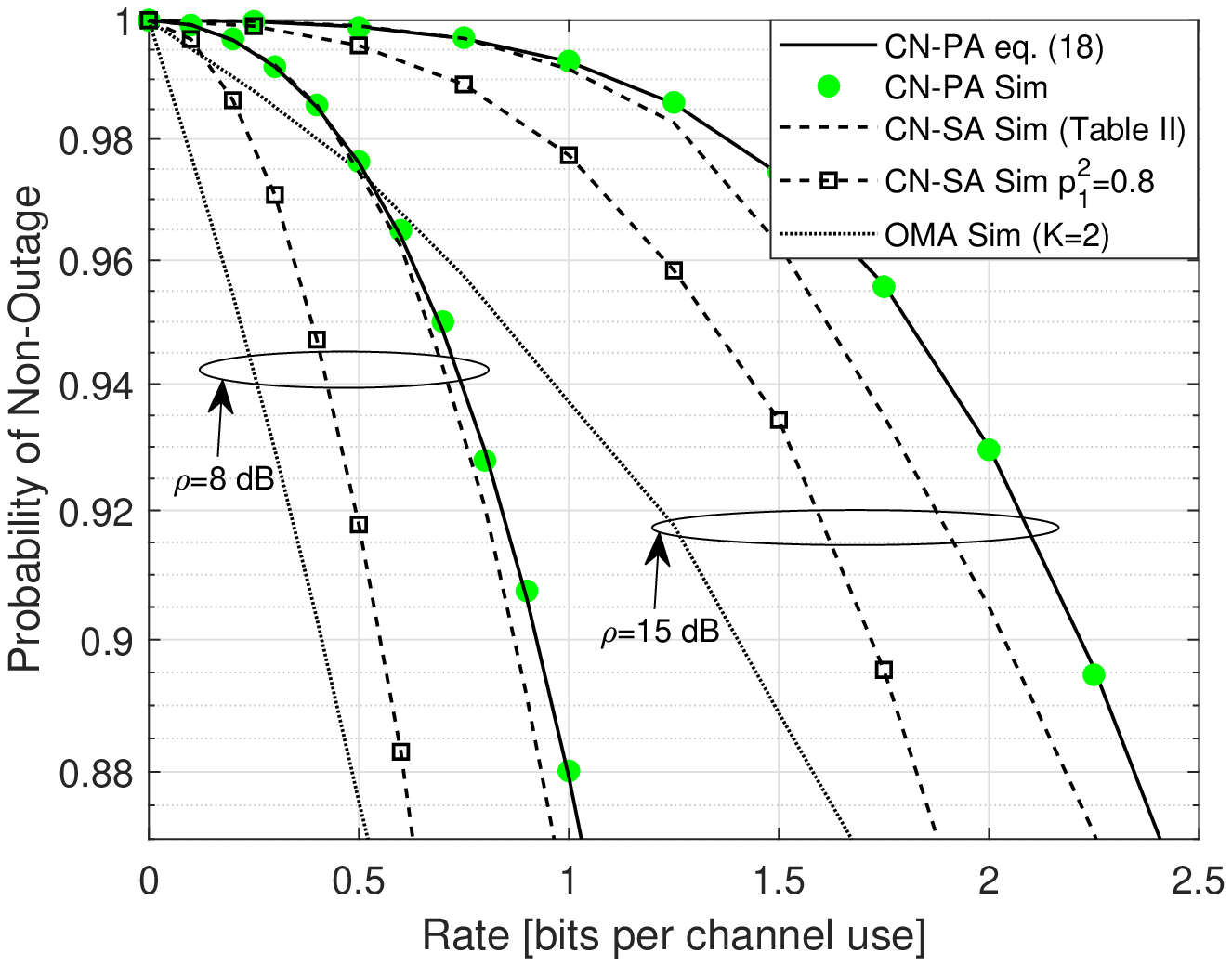}
		\caption{Outage Capacity: Rate supported by the CO-NOMA schemes against probability of non-outage, where $K=2$ and $R_1=R_2$.}
		\label{fig:VaryRatePlotOutage8dB15dB}
	\end{figure}

		\figref{fig:VaryRatePlotOutage8dB15dB} compares the outage capacities achieved by the schemes. For a given target rate, the curve specifies the probability of non-outage. It can be observed that the rate supported by the CN-PA scheme is higher compared to both CN-SA and OMA. Infact, in a practical non-power-balanced scenario (i.e. CN-SA with $p_1^2=0.8$) with 90\% non-outage, CN-PA is able to support approximately 30\% higher transmission rate at 15 dB, compared to CN-SA. 
	
	\section{Conclusion}
	A heuristic power allocation scheme for CO-NOMA that improves the transmission reliability has been proposed. Numerical simulations show that the proposed scheme achieves a lower MOP and supports a higher transmission rate as compared to the conventional scheme. Moreover, the exact closed-form analytical expressions for MOP for the two user case are derived and it is shown that each user is able to achieve full diversity.
	
	\section*{Appendix A}
	\comment{
		time slot 1: BS broadcasts. The received signal at the user with poor
		channel in the first time slot is 
		
		\[
		y_{1}=\sqrt{\rho}p_{1}h_{1}s_{1}+\sqrt{\rho}p_{2}h_{1}s_{2}+n_{1}
		\]
		where $|h_{1}|^{2}<|h_{2}|^{2}$ and received signal at the user with
		better channel in the first time slot is 
		
		\[
		y_{2}=\sqrt{\rho}p_{1}h_{2}s_{1}+\sqrt{\rho}p_{2}h_{2}s_{2}+n_{2}
		\]
		Power allocation: $|h_{1}|^{2}<|h_{2}|^{2}$ and $p_{1}<p_{2}$ implies
		best channel gets higher power and poor channel gets more copies!
		Best channel user will decode first and forward decoded data to farthest
		user with poor channel. Poor channel user gets advantage of multiple
		signal copies. Best channel user gets advantage of high power. So
		best user will decode signal first using SIC. It will decode user
		1 data and then decode its own 
		
		\textbf{SNRs at user 2 (Best user); }
	}
	
	For the two user case, MOP is given as 
	\begin{equation}
	\label{eq:A1m}
	\mathcal{M_O}=1-\text{Pr}\left(\gamma_{1}>\phi_{1},\gamma_{2}>\phi_{2}\right).
	\end{equation}
	Substituting the values of SINR from \eqref{eq:SINRuser2} and \eqref{eq:SINRuser1} yields 
	\comment{
		\[
		MOP=1-Pr\left(\rho p_{l}|h_{1}|^{2}+\rho|g|^{2}>\phi_{1},\frac{p_{h}|h_{2}|^{2}}{p_{l}|h_{2}|^{2}+\frac{1}{\rho}}>\phi_{2}\right)
		\]
		\[
		MOP=1-Pr\left(|h_{1}|^{2}>\frac{1}{p_{l}}\left(\frac{\phi_{1}}{\rho}-\delta\right),|h_{2}|^{2}>\frac{1}{\rho p_{h}}\frac{\phi_{2}}{\left(1-\phi_{2}\frac{P_{l}}{p_{h}}\right)}\right)
		\]
	}
	\begin{equation}
	\mathcal{M_O}=1-\text{Pr}\left(|h_{1}|^{2}>\alpha,|h_{2}|^{2}>\beta \right),
	\end{equation}
	where $\alpha=\frac{1}{p_{l}}\left(\frac{\phi_{1}}{\rho}-\delta\right)$, $\beta=\frac{1}{\rho p_{h}}\frac{\phi_{2}}{\left(1-\phi_{2}\frac{P_{l}}{p_{h}}\right)}$
	and $\delta=|g|^{2}$. MOP can be further expressed as
	\comment{
		\[
		MOP=\begin{cases}
		1-Pr\left(|h_{2}|^{2}>\frac{1}{\rho p_{h}}\frac{\phi_{2}}{\left(1-\phi_{2}\frac{P_{l}}{p_{h}}\right)}\right) & \delta>\frac{\phi_{1}}{\rho}\\
		1-Pr\left(|h_{1}|^{2}>\frac{1}{p_{l}}\left(\frac{\phi_{1}}{\rho}-\delta\right),|h_{2}|^{2}>\frac{1}{\rho p_{h}}\frac{\phi_{2}}{\left(1-\phi_{2}\frac{P_{l}}{p_{h}}\right)}\right) & \delta<\frac{\phi_{1}}{\rho}
		\end{cases}
		\]
	}
	\begin{equation}
	\label{eq:MOPA0}
	\mathcal{M_O}=\begin{cases}
	\mathcal{M_{O,G}}=1-\text{Pr}\left(|h_{2}|^{2}>\beta\right) & \delta>\frac{\phi_{1}}{\rho}\\
	\mathcal{M_{O,L}}=	1-\text{Pr}\left(|h_{1}|^{2}>\alpha,|h_{2}|^{2}>\beta\right) & \delta<\frac{\phi_{1}}{\rho}.
	\end{cases}
	\end{equation}
	\emph{Expression for $\mathcal{M_{O,G}}$:}
	\begin{equation}
	\label{eq:MOPA1}
	\mathcal{M_{O,G}}=1-\text{Pr}\left(|h_{2}|^{2}>\beta\right)=1-\left(1-e^{-\frac{1}{\rho
			p_{h}}\frac{\phi_{2}}{\left(1-\phi_{2}\frac{P_{l}}{p_{h}}\right)}}\right)^{2}.
	\end{equation}
	\emph{Expression for $\mathcal{M_{O,L}}$:}  Let the unsorted channel power gains be denoted by $\mu$ and $\kappa$, which implies $|h_{2}|^{2}=\max\left\{ \mu,\kappa\right\} $ and $|h_{1}|^{2}=\min\left\{\mu,\kappa\right\} $. $\mathcal{M_{O,L}}$ can be expressed as
	\begin{equation}
	\mathcal{M_{O,L}}=1-\text{Pr}\left(\min\left\{ \mu,\kappa\right\} >\alpha,\max\left\{ \mu,\kappa\right\} >\beta\right).
	\end{equation}
	$\mathcal{M_{O,L}}$ can be further expressed as
	\begin{equation*}
	\begin{split}
	\mathcal{M_{O,L}}=1 -\left(\text{Pr}\left(\delta<\frac{\phi_{1}}{\rho},\mu>\alpha,\kappa>\beta,\kappa>\mu\right)
	+\text{Pr}\left(\delta<\frac{\phi_{1}}{\rho},\kappa>\alpha,\mu>\beta,\mu>\kappa\right)\right).
	\end{split}
	\end{equation*}
	The $\mathcal{M_{O,L}}$ expression can be further expanded as
	\begin{equation}
	\label{eq:MOPProbform2}
	\begin{split}
	\mathcal{M_{O,L}}=1&-\text{Pr}\left(\delta<\frac{\phi_{1}}{\rho},\alpha>\beta,\alpha<\mu<\kappa\right)
	\\
	&-\text{Pr}\left(\delta<\frac{\phi_{1}}{\rho},\alpha>\beta,\alpha<\kappa<\mu\right)
	\\
	&-\text{Pr}\left(\delta<\frac{\phi_{1}}{\rho},\alpha<\beta,\alpha<\mu<\kappa,\beta<\kappa\right)
	\\
	&-\text{Pr}\left(\delta<\frac{\phi_{1}}{\rho},\alpha<\beta,\alpha<\kappa<\mu,\beta<\mu\right).
	\end{split}
	\end{equation}
	Condition $(\alpha>\beta)$, is equivalent to $(\delta<c)$, where
	$\scriptstyle c=\frac{1}{\rho}\left(\phi_{1}-\frac{\phi_{2}}{\left(\frac{p_{h}}{p_{l}}-\phi_{2}\right)}\right)$.
	This implies, the condition $\scriptstyle\left\{ \delta<\frac{\phi_{1}}{\rho}\cap\alpha>\beta\right\} $
	can be combined as $\scriptstyle\left\{ \delta<\Delta\right\} $, where $\scriptstyle\Delta=\min\left\{ c,\frac{\phi_{1}}{\rho}\right\} $ and condition
	$\scriptstyle\left\{ \delta<\frac{\phi_{1}}{\rho}\cap\alpha<\beta\right\} $ is equivalent to $\scriptstyle\left\{
	c<\delta<\frac{\phi_{1}}{\rho}\right\} $.
	The resulting expression of $\mathcal{M_{O,L}}$ is expressed as
	\comment{
		\[
		MOP=1-\left(Pr\left(\delta<\Delta,\alpha<\mu<\kappa\right)+Pr\left(\delta<\Delta,\alpha<\kappa<\mu\right)+Pr\left(c<\delta<\frac{\phi_{1}}{\rho},\alpha<\mu<\kappa,\beta<\kappa\right)+Pr\left(c<\delta<\frac{\phi_{1}}{\rho},\alpha<\kappa<\mu,\beta<\mu\right)\right)
		\]
		This can be expressed as
	}
	\begin{equation}
	\label{eq:MOPA2}
	\begin{split}
	\mathcal{M_{O,L}}=1-\left(\mathcal{P}_{1}+\mathcal{P}_{2}+\mathcal{P}_{3}+\mathcal{P}_{4}\right),
	\end{split}
	\end{equation}
	where
	$\mathcal{P}_{1}=\text{Pr}\left(\delta<\Delta,\alpha<\mu<\kappa\right)$,
	$\mathcal{P}_{2}=\text{Pr}\left(\delta<\Delta,\alpha<\kappa<\mu\right)$,
	$\mathcal{P}_{3}=\text{Pr}\left(c<\delta<\frac{\phi_{1}}{\rho},\alpha<\mu<\kappa,\beta<\kappa\right)$
	and $\mathcal{P}_{4}=\text{Pr}\left(c<\delta<\frac{\phi_{1}}{\rho},\alpha<\kappa<\mu,\beta<\mu\right)$.
	$\mathcal{P}_{1}$ can be obtained as
	\begin{equation*}
	\begin{split}
	\mathcal{P}_{1}=\text{Pr}\left(\delta<\Delta,\alpha<\mu<\kappa\right)=\int_{0}^{\Delta}\int_{\alpha}^{\infty}\int_{x}^{\infty}f_{\mu,\kappa,\delta}\left(x,y,z\right)dydxdz.
	\end{split}
	\end{equation*}
	Substituting the joint PDF, $ f_{\mu,\kappa,\delta}\left(x,y,z\right)=e^{-\left(x+y+z\right)}u(x)u(y)u(z)$, and solving the resulting integral yields
	\comment{
		\[
		\mathcal{P}_{1}=\int_{0}^{\Delta}\int_{\alpha}^{\infty}\int_{x}^{\infty}e^{-\left(x+y+z\right)}dydxdz
		\]
		Solving yields,
	}
	\begin{equation}
	\mathcal{P}_{1}=\frac{1}{2}\frac{e^{-2\left(\frac{1}{p_{l}}\frac{\phi_{1}}{\rho}\right)}}{\left(1-\frac{2}{p_{l}}\right)}\left(1-e^{-\left(1-\frac{2}{p_{l}}\right)\Delta}\right).
	\end{equation}
	Similarly, $\mathcal{P}_{2}=\int_{0}^{\Delta}\int_{\alpha}^{\infty}\int_{\alpha}^{x}f_{\mu,\kappa,\delta}\left(x,y,z\right)dydxdz$ is solved to yield
	\begin{equation}
	\begin{split}
	\mathcal{P}_{2}=\left(\frac{1}{2}\frac{e^{-\frac{2}{p_{l}}\frac{\phi_{1}}{\rho}}}{\left(1-\frac{2}{p_{l}}\right)}\left(1-e^{-\left(1-\frac{2}{p_{l}}\right)\Delta}\right)\right).
	\end{split}
	\end{equation}
	It can be noted that $\mathcal{P}_{2}=\mathcal{P}_{1}$, because the distribution is symmetric.
	$\mathcal{P}_{3}=\int_{c}^{\frac{\phi_{1}}{\rho}}\int_{\alpha}^{\infty}\int_{\beta}^{x}e^{-\left(x+y+z\right)}dydxdz$ is solved to yield
	\begin{equation*}
	\begin{split}
	\mathcal{P}_{3}
	=e^{-\left(\beta+\frac{1}{p_{l}}\frac{\phi_{1}}{\rho}\right)}\left(\frac{e^{-c\left(1-\frac{1}{p_{l}}\right)}-e^{-\frac{\phi_{1}}{\rho}\left(1-\frac{1}{p_{l}}\right)}}{\left(1-\frac{1}{p_{l}}\right)}\right)-\frac{e^{-2\beta}}{2}\left(e^{-c}-e^{-\frac{\phi_{1}}{\rho}}\right).
	\end{split}
	\end{equation*}
	$\mathcal{P}_{4}=\int_{c}^{\frac{\phi_{1}}{\rho}}\int_{\beta}^{\infty}\int_{\alpha}^{x}f_{\mu,\kappa,\delta}\left(x,y,z\right)dydxdz$  is solved to yield
	\begin{equation*}
	\begin{split}
	\mathcal{P}_{4}=\left(\frac{e^{-\beta}e^{-\frac{1}{p_{l}}\frac{\phi_{1}}{\rho}}}{\left(1-\frac{1}{p_{l}}\right)}\left(e^{-c\left(1-\frac{1}{p_{l}}\right)}-e^{-\frac{\phi_{1}}{\rho}\left(1-\frac{1}{p_{l}}\right)}\right)-\frac{e^{-2\beta}}{2}\left(e^{-c}-e^{-\frac{\phi_{1}}{\rho}}\right)\right).
	\end{split}
	\end{equation*}
	Again, in this case, $\mathcal{P}_{4}=\mathcal{P}_{3}$. Using \eqref{eq:MOPA0}, \eqref{eq:MOPA1} and \eqref{eq:MOPA2}, the final expression for MOP is obtained as \eqref{eq:MOPFinal}.

		\section*{Appendix B }
		
		The MOP of a two user network is defined in \eqref{eq:A1m}. It is trivial to show that 
		\begin{equation}
		\label{eq:A2m}
		\int_{x>\phi} f_{\gamma_i} (x) dx \geq  \int_{y>\phi} \int_{x>\phi} f_{\gamma_1,\gamma_2} (x,y) dx dy, 
		\end{equation}	 
		where $i=\{1,2\}$. \eqref{eq:A2m} implies
		\begin{equation}
		\label{eq:A3m}
		\text{Pr}\left(\gamma_i > \phi\right) \geq \text{Pr}\left(\gamma_{1}>\phi,\gamma_{2}>\phi\right).
		\end{equation}	 
		From \eqref{eq:A1m} and \eqref{eq:A3m}, we can deduce that $\text{Pr}\left(\gamma_{i}<\phi\right)=1-\text{Pr}\left(\gamma_{i}>\phi\right)\leq \mathcal{M_O}$ i.e. outage probability of the respective users will always be less than or equal to the MOP.

	\bibliography{IEEEfull,document}
	\bibliographystyle{IEEEtran}

\end{document}